\documentclass[12pt]{article}

\usepackage{amssymb}
\usepackage[T1]{fontenc}
\usepackage[utf8]{inputenc}
\usepackage[intlimits]{amsmath}
\usepackage{enumerate}
\usepackage[all]{xy}

\usepackage{amsthm}

\newcommand{\be}{\begin{equation}}
\newcommand{\ee}{\end{equation}}

\newtheorem{theorem}{Theorem}
\newtheorem{corollary}{Corollary}

\newtheorem{example}{Example}

\usepackage{multirow}

\usepackage[colorlinks]{hyperref}

\begin{document}

\begin{center}\Large\bf
A new look at Lie algebras
\end{center}
\begin{center}\large
Alina Dobrogowska, Grzegorz Jakimowicz
\end{center}

\begin{center}
Faculty of Mathematics, University of Białystok, Ciołkowskiego 1M, 15-245 Białystok, Poland
\end{center}
\begin{center}
E-mail: alina.dobrogowska@uwb.edu.pl, g.jakimowicz@uwb.edu.pl
\end{center}

\begin{abstract}
We present a new look at description of real finite-dimensional Lie algebras. 
The basic element turns out to be a pair $(F,v)$ consisting of a linear mapping $F\in End(V)$ and its eigenvector $v$.
This pair allows to build a Lie bracket on a dual space to a linear space $V$. 
This algebra is solvable. In particular, when $F$ is nilpotent, the Lie algebra is also nilpotent. 
We show that these solvable algebras are the basic bricks of the construction of all other Lie algebras.
Using relations between the Lie algebra, the Lie--Poisson structure and the Nambu bracket, we show that the algebra invariants (Casimir functions) are solutions of an equation which has a geometric sense.
Several examples illustrate the importance of these constructions.
\end{abstract}


\section{Introduction}

There is a well-known isomorphism between special orthogonal Lie algebra $\mathfrak{so}(3)$ and $\mathbb{R}^3$.
For the first structure, the Lie bracket is given by the matrix commutator $[X,Y]=XY-YX$ for $X,Y\in \mathfrak{so}(3)$, and for the second by the cross product $\times$ for vectors from $\mathbb{R}^3$. The mapping 
\begin{equation}
X=\left(
\begin{matrix}
0 & -z & y\\
z & 0 & -x\\
-y & x & 0
\end{matrix}
\right)\longmapsto 
v=\left(\begin{matrix} x\\y\\z \end{matrix}\right)
\end{equation}
gives this isomorphism $\left(\mathfrak{so}(3), [\cdot,\cdot]\right)\cong \left(\mathbb{R}^3,\times \right)$.

The main goal of the paper is to show that one can construct a similar isomorphism for any Lie algebra.
We will show that Lie algebras have a lot in common with linear maps, and more precisely with linear maps with a fixed eigenvector. 
We will present an easy construction of a Lie bracket on a linear space $V^*$ from a pair $(F,v)$, where $F\in End(V)$ and $v$ is an eigenvector of $F$. In our considerations, we will restrict ourselves to the linear space $V$ over a field $\mathbb{R}$. It means that we will analyze in detail only real Lie algebras. However, we want to emphasize that the presented formulas also work for vector spaces over the field of complex numbers.

The presented construction seems to us new and at the same time very simple. 
 Moreover, it generates natural connections between algebra, analysis and geometry.
To be precise, we have relationship between objects such as: Lie algebras, eigenproblem for operators, Lie--Poisson brackets, Nambu brackets and finally invariants (Casimir operators) for initial Lie algebras.

The idea of constructing these Lie brackets has its origin in the algebroid structure. The concept of Lie algebroid come from Lie groupoids \cite{Pra}. A Lie algebroid $(A,[\cdot,\cdot],a)$ is a vector bundle $A\longrightarrow M$ over manifold $M$, together with a vector bundle map (anchor map) $a:A\longrightarrow TM$ and Lie bracket $[\cdot,\cdot]:\Gamma (A)\times \Gamma (A)\longrightarrow \Gamma (A)$.
The anchor and  the Lie bracket satisfy the Leibniz rule
\begin{equation}
[\alpha,f\beta]=f[\alpha,\beta]+a(\alpha)(f)\beta
\end{equation}
for all $\alpha,\beta\in  \Gamma (A)$, $f\in C^{\infty}(M)$.
It is a well-known fact that there are one-to-one correspondences between vector bundles with linear Poisson structures and Lie algebroids as well as between vector spaces with linear Poisson structures and Lie algebras.
There are many books and articles about  Lie albebroids or related structures \cite{13, AJ, Zu, GraUrb, J-R, Kos, 2, Ma-O-S, 3, 4, O-J-S, We, Xu}.

The  idea of Lie brackets construction  has  exactly its source in our work \cite{Do-Ja} (see also \cite{AJK1, AJK3}), in which we presented new ideas how to construct algebroid structures on cotangent bundles starting from classical algebroid structures on $TM$ and vector fields satisfying certain assumptions. 
We also found in this paper constructions of families of algebroids on $A^{*}$ when we have an algebroid structure on $A$. 
It is also possible to look at these brackets in the language torision-free and flat connections as described by Balcerzak \cite{Bal}.

It also seems that linking the Lie bracket to eigenproblems of linear operators may be a step towards the classification of Lie algebras. The problem of classifying real Lie algebras is completely solved for Lie algebras up to dimension six, see \cite{Bo-Fe, Ci-Gr-Sch, G, Gr-2, Mac, Mu, Po-Bo-Ne} and the references therein.

The paper is organized as follows. In the beginning of Section \ref{s2} we  recall basic facts about Lie algebras, Poisson brackets and Nambu brackets.
Next sections contain the main results of the article. In Section \ref{s3} we define the Lie bracket on $V^*$ starting from a pair $(F,v)$, where $F\in End(V)$ and $v\in V$ (Theorem \ref{theorem 1}). Moreover, we study the properties of Lie algebras obtained in this way. We present some examples in Section \ref{s4}. The culmination of our research are Theorem \ref{theorem-1} and \ref{theorem-2} in Section \ref{s5},
where we show that every Lie algebra and its invariants can be associated with a collection of linear mappings and their eigenvectors.

\section{Preliminaries and notations}
\label{s2}

In the present section  we give a short review of some basic definitions and properties about Lie algebras, Poisson structures, Nambu brackets.

Let $\mathfrak{g}$ be a real finite-dimensional Lie algebra. A Lie algebra is vector space over a field $\mathbb{R}$ equipped with Lie bracket $[\cdot, \cdot]:\mathfrak{g}\times \mathfrak{g}\longrightarrow \mathfrak{g}$ with is a bilinear, antisymmetric map, which satisfies the Jacobi identity
\begin{equation}
[[x,y],z]+[[z,x],y]+[[y,z],x]=0
\end{equation}
for all $x,y,z\in\mathfrak{g}$.

We say that a linear subspace $\mathfrak{h}$ is an ideal of a Lie algebra $\mathfrak{g}$ when $[\mathfrak{g},\mathfrak{h}]\subseteq  \mathfrak{h}$. Of course the set $[\mathfrak{h},\mathfrak{h}]$ is also an ideal. In particular, $\mathfrak{g}$ is a trivial ideal of $\mathfrak{g}$. Then we define a sequence of ideals (the derived series $\mathfrak{g}^{(0)}\supseteq \mathfrak{g}^{(1)} \supseteq \dots\supseteq \mathfrak{g}^{(i)}\supseteq \dots $)
\begin{equation}
\mathfrak{g}^{(0)}=\mathfrak{g},\quad \mathfrak{g}^{(1)}=[\mathfrak{g},\mathfrak{g}],\quad 
\mathfrak{g}^{(2)}=[\mathfrak{g}^{(1)},\mathfrak{g}^{(1)}], \dots,  
\mathfrak{g}^{(i)}=[\mathfrak{g}^{(i-1)},\mathfrak{g}^{(i-1)}], \dots
\end{equation}
A Lie algebra $\mathfrak{g}$ is called solvable if, for some positive integer $i$, $\mathfrak{g}^{(i)}=0$.
In addition, if we introduce the following sequence of ideals (the lower central series $\mathfrak{g}_{(0)}\supseteq \mathfrak{g}_{(1)} \supseteq \dots\supseteq \mathfrak{g}_{(i)}\supseteq \dots $)
\begin{equation}
\mathfrak{g}_{(0)}=\mathfrak{g},\quad \mathfrak{g}_{(1)}=[\mathfrak{g}_{(0)},\mathfrak{g}],\quad 
\mathfrak{g}_{(2)}=[\mathfrak{g}_{(1)},\mathfrak{g}], \dots,  
\mathfrak{g}_{(i)}=[\mathfrak{g}_{(i-1)},\mathfrak{g}], \dots,
\end{equation}
we say that algebra $\mathfrak{g}$ is called nilpotent if the lower central series terminates $\mathfrak{g}_{(i)}=0$ for some $i\in\mathbb{N}$. Obviously, a nilpotent Lie algebra is also solvable.

Let $M$ be a finite-dimensional smooth manifold. A Poisson bracket $\{\cdot,\cdot\}:C^{\infty}(M)\times C^{\infty}(M)\longrightarrow  C^{\infty}(M)$ is a bilinear, antisymmetric map, which satisfies the Jacobi identity and Leibniz rule 
\begin{align}
& \{\{f,g\},h\}+\{\{h,f\},g\}+\{\{g,h\},f\}=0,\\
& \{f,gh\}=g\{f,h\}+\{f,g\}h
\end{align}
for $f,g,h\in C^{\infty}(M)$. A pair $(M, \{\cdot,\cdot\})$ is called a Poisson manifold. 

One classic example of such a structure is the Lie-Poisson bracket on the dual $\mathfrak{g}^*$ of a finite-dimensional Lie algebra $\mathfrak{g}$
\begin{equation}
\label{bp}
\{f,g\}({\bf x})=\left<{ \bf x}|[df({\bf x}),dg({\bf x})]\right>,\quad f,g\in C^{\infty}\left(\mathfrak{g}^*\right),
\end{equation}
where $df({\bf x}),dg({\bf x})\in \left(\mathfrak{g}^*\right)^*\cong \mathfrak{g}$ and ${\bf x}\in\mathfrak{g}^*$.
It is well known that we have a one-to-one correspondence between the Lie algebra structures on $\mathfrak{g}$ and the linear Poisson structures on $\mathfrak{g}^*$. Let $\{e_1,\dots, e_n\}$ be a basis of $\mathfrak{g}$ and ${\bf x}=(x_1,\dots, x_n)$ be a system of local coordinates on the dual space.
Using the structure constants $c_{ij}^k$ of the Lie algebra $[e_i,e_j]=\sum_{k=1}^n c_{ij}^ke_k$, one can express the Lie-Poisson bracket as a linear function $\{x_i,x_j\}=\sum_{k=1}^n c_{ij}^k x_k$.

Recall also that Casimir functions $c_i$, $i=1,\dots ,k$, for Poisson manifold $(M, \{\cdot,\cdot\})$ are defined by the condition
\begin{equation}
\{c_i,f\}=0 \quad \textrm{for all}\quad f\in C^{\infty}(M).
\end{equation}
In the linear case for the Poisson bracket, these functions correspond to invariants (Casimir operators) of the Lie algebra.

In 1973, Nambu \cite{Nambu} proposed a generalization of the Poisson bracket on $\mathbb{R}^3$ to the Nambu bracket in the form
\begin{equation}
\{f_1,f_2,f_3\}({\bf x})=\dfrac{\partial(f_1,f_2,f_3)}{\partial (x_1,x_2,x_3)}=\sum_{i,j,k=1}^3 
\epsilon_{ijk} \dfrac{\partial f_1}{\partial x_i}({\bf x}) \dfrac{\partial f_2}{\partial x_j}({\bf x}) \dfrac{\partial f_3}{\partial x_k} ({\bf x}),
\end{equation}
where $\epsilon $ is Levi-Civita tensor and $f_1,f_2,f_3\in C^{\infty}(\mathbb{R}^3)$.
 In general, a Nambu bracket
$\{\cdot,\dots,\cdot\}:\underbrace{C^{\infty}(M)\times \dots \times C^{\infty}(M)}_n\longrightarrow  C^{\infty}(M)$ is a $n$-linear, skew-symmetric map, which satisfies the generalized Jacobi identity (fundamental identity) 
\begin{equation}
\{f_1, \dots, f_{n-1},\{g_1,\dots ,g_n\}\}=\sum_{i=1}^n \{g_1, \dots,  \{f_1, \dots, f_{n-1},g_i\}, \dots, g_n\}
\end{equation}
and Leibniz rule (derivation law)
\begin{equation}
\{f_1, \dots, f_{n-1},fg\}=f\{f_1, \dots, f_{n-1},g\}+\{f_1, \dots, f_{n-1},f\}g.
\end{equation}
More information about Nambu structure can be found in \cite{A-L-M-Y, Ch-Ho, Cr-Mo, Gra-Mar, Ho-Ma, Sh, Tak}.

\section{Eigenproblems of operators and Lie brackets}
\label{s3}

In this section we present some constructions of 
a Lie bracket on a space $V^*$ having a pair: linear mapping and its eigenvector. 
We shall examine the properties of the Lie brackets obtained in this way.
We consider when different mappings (and eigenvectors) yield isomorphic structures of Lie algebras on $V^*$. 
Finally, we give a criterion when the linear combination of two brackets generated by two different mappings and their eigenvectors gives again the Lie bracket.

\begin{theorem}
\label{theorem 1}
 If $V$ is a vector space, $F:V \longrightarrow  V$ is a linear map and $ v\in V $ is an eigenvector of the map $F$, then 
$(V^*,[\cdot,\cdot]_{F,v})$, is a Lie algebra, where the Lie bracket is given by
 \begin{equation}
 \label{b1} [\psi,\phi]_{F,v}= \phi(v)F^*(\psi)-\psi(v)F^*(\phi)
\end{equation}
 for $\psi,\phi\in V^*$.
\end{theorem}

\begin{proof}
Let $v\in V$ be an eigenvector of a linear map $F$ corresponding to the eigenvalue $\lambda$, 
$F(v)=\lambda v$. The bilinearity and antisymmetry  for the bracket (\ref{b1}) are obvious. All we have to prove is the Jacobi identity. 
Let us take $\psi,\phi,\zeta \in V^*$. We calculate
\begin{equation}[[\psi,\phi]_{F,v},\zeta]_{F,v}=\psi(v)\left(F^*(\phi)(v)F^*(\zeta)-\zeta(v)F^*(F^*(\phi))\right)
\end{equation}
 $$
-\phi(v)\left( F^*(\psi)(v)F^*(\zeta)-\zeta(v)F^*(F^*(\psi))\right)$$
$$=\phi(v)\zeta(v)F^*(F^*(\psi))-\psi(v)\zeta(v)F^*(F^*(\phi)),$$
because $F^*(\phi)(v)=\phi (F(v))=\lambda \phi (v)$. Hence 
\begin{equation}
[[\psi,\phi]_{F,v},\zeta]_{F,v}+[[\zeta,\psi]_{F,v},\phi]_{F,v}+[[\phi,\zeta]_{F,v},\psi]_{F,v}
\end{equation}
 $$=\phi(v)\zeta(v)F^*(F^*(\psi))-\psi(v)\zeta(v)F^*(F^*(\phi))+
\phi(v)\psi(v)F^*(F^*(\zeta)) $$
$$-\phi(v)\zeta(v)F^*(F^*(\psi))+\psi(v)\zeta(v)F^*(F^*(\phi))
-\psi(v)\phi(v)F^*(F^*(\zeta))=0.$$
\end{proof}
If we have two linear mappings $F$ and $G$ that have a common eigenvector $v$ then the following property holds
$[\psi,\phi]_{F,v}+[\psi,\phi]_{G,v}=[\psi,\phi]_{F+G,v}$. 

By introducing the notion of the internal multiplication operator $\iota_v:V^*\longrightarrow \mathbb{R}$ given by $\iota_v\psi=\psi(v)$, 
we can rewrite the Lie bracket (\ref{b1}) in the form
\begin{equation}
[\psi,\phi]_{F,v}= \left(\iota_v\wedge F^*\right) \left(\phi , \psi \right).
\end{equation}
Since we have the canonical isomorphism $V^{**}\cong V$, we will identify the vector $v$ with $\iota_v$.
Moreover, we will use isomorphisms $End(V)\cong V^*\otimes V$ or $End(V^*)\cong V\otimes V^*$ defined by
$(\psi\otimes v)(w)=\psi(w) v$, $(v\otimes \psi)(\phi)=\phi(v)\psi$ for $\psi,\phi\in V^*$, $v,w\in V$, respectively.
These formulas will be useful in the next theorems.

\begin{theorem}
 If $F:V \longrightarrow  V$ is a linear map, $v\in V$ and 
$(V^*,[\cdot,\cdot]_{F,v})$ is a Lie algebra,  then one of the following conditions is satisfied:
\begin{enumerate}
\item $v$ is an eigenvector of the map $F$;
\end{enumerate}
or
\begin{enumerate}
\item[2.] $F^*=\iota_v \otimes \rho  + \iota_{Fv} \otimes  \eta$, where $\rho,\eta\in V^*$.
\end{enumerate}
\end{theorem}

\begin{proof}
For $F=0$ it is obvious. Suppose now that $F$ is not equal zero. 
After simple calculation we obtain  
\begin{equation}
\label{77}
[[\psi,\phi]_{F,v},\zeta]_{F,v}+[[\zeta,\psi]_{F,v},\phi]_{F,v}+[[\phi,\zeta]_{F,v},\psi]_{F,v}=
\end{equation}
$$=\psi(v)\phi(Fv)F^*(\zeta)+\zeta(Fv)\phi(v)F^*(\psi)+\psi(Fv)\zeta(v)F^*(\phi)$$
$$-\psi(Fv)\phi(v)F^*(\zeta)-\zeta(v)\phi(Fv)F^*(\psi)-\psi(v)\zeta(Fv)F^*(\phi)$$
$$= \bigg( \iota_v \wedge \iota_{Fv} \wedge F^* \bigg) (\psi,\phi,\zeta).$$
This implies the above subcases. This ends the proof.
\end{proof}

{\bf Remark 1.}
We will study the second possibility in more detail.
Suppose that  vectors $v$ and $Fv$ are linearly independent. Then let  $\{e_1=v,e_2=Fv, e_3,...,e_n\}$  be a basis of $V$ and  $\{e_1^*,e_2^*, e_3^*,...,e_n^*\}$ be the dual basis in $V^*$. If we put $\psi=e_1^*, \phi=e_2^*, \zeta =e_i^*$ we obtain from (\ref{77})
\begin{equation}
F^*(e_i^*-e_i^*(e_1)e_1^*-e_i^*(e_2)e_2^*)=0.
\end{equation}
It implies that $F^*(e_i^*)=0$ for $i=3,...,n$. 
Therefore, the mapping in this basis has the form
\begin{equation}
F=\left(
\begin{matrix}
0 & a_{12} & a_{13} & \hdots & a_{1n}\\
1 & a_{22} & a_{23} & \hdots & a_{2n}\\
0 & 0      & 0      & \hdots & 0\\
\vdots & \vdots & \vdots & \ddots & \vdots\\
0 & 0      & 0      & \hdots & 0
\end{matrix}
\right),
\end{equation}
where $a_{1i}, a_{2i}\in\mathbb{R}$, $i=2,\dots, n$. In this case the only non-zero bracket of basis elements for bracket (\ref{b1}) is the following
\begin{equation}
[e_1^*,e_2^*]_{F,e_1}=-e_1^*-a_{22}e_2^*-\ldots -a_{2n}e_n^*.
\end{equation}
We recognize that it is the Lie algebra structure related to the Lie algebra $\mathfrak{g}_{2,1}$ (more exactly it is isomorphic with the direct sum $\mathfrak{g}_{2,1}\oplus \left< e_3^*,...,e_n^*\right>$). An isomorphism is given by the mapping $\{e_1^*,e_2^*, e_3^*,...,e_n^*\}\mapsto \{e_2^*,e_1^*+a_{22}e_2^*+\ldots +a_{2n}e_n^*, e_3^*,...,e_n^*\}$. 

Even though this algebra was obtained using vector $v$ which was not an eigenvector of $F$, it can also be obtained from the mapping $F=diag (0,-1,0\dots, 0)$ and the eigenvector $v=e_1$. Thus without the loss of generality from now on we will always assume that $v$ is an eigenvector for $F$.

\bigskip

It is easy to prove that 
this construction of Lie brackets gives solvable Lie algebras.
\begin{theorem}
Let $[\cdot,\cdot]_{F,v}$ be given by (\ref{b1}), then the Lie algebra $(V^*,[\cdot,\cdot]_{F,v})$ is solvable.
\end{theorem}

\begin{proof}
To prove above, we note that
\begin{equation}
[[\psi,\phi]_{F,v},[\varphi ,\zeta  ]_{F,v}]_{F,v}
\end{equation}
$$
=\phi(v)\zeta(v)\left( \varphi(F(v))(F^*)^2(\psi) - \psi(F(v))(F^*)^2(\varphi) \right)
$$$$
-\phi(v)\varphi(v)\left( \zeta(F(v))(F^*)^2(\psi) - \psi(F(v))(F^*)^2(\zeta) \right)
$$$$
-\psi(v)\zeta(v)\left( \varphi(F(v))(F^*)^2(\phi) - \phi(F(v))(F^*)^2(\varphi) \right)
$$$$
+\psi(v)\varphi(v)\left( \zeta(F(v))(F^*)^2(\phi) - \phi(F(v))(F^*)^2(\zeta) \right)
=0,
$$
for all $\psi,\phi,\varphi,\zeta \in V^*=\mathfrak{g}$.
It means that
$
\mathfrak{g}^{(2)}=[[\mathfrak{g},\mathfrak{g}]_{F,v},[\mathfrak{g},\mathfrak{g}]_{F,v}]_{F,v}=0.
$
\end{proof}

Additionally, if we assume that the mapping $F$ giving the Lie bracket is nilpotent, then we obtain a nilpotent Lie algebra.
\begin{theorem}
If $F$ is a nilpotent operator, then  $(V^*,[\cdot,\cdot]_{F,v})$ is a nilpotent Lie algebra.
\end{theorem}

\begin{proof}
Denote $\mathfrak{g}=V^*$.
As a first step, we calculate
\begin{equation}
[[\psi,\phi]_{F,v},\varphi]_{F,v}=[\phi(v)F^*(\psi)-\psi(v)F^*(\phi),\varphi]_{F,v}
\end{equation}
$$
=\varphi(v)\left( \phi(v)(F^*)^2(\psi)-\psi(v)(F^*)^2(\phi)\right),
$$
because $Fv=\lambda v$. Thus
\begin{equation}
[[[\psi,\phi]_{F,v},\varphi]_{F,v},\zeta]_{F,v}=\zeta(v)\varphi(v)\left( \phi(v)(F^*)^3(\psi)-\psi(v)(F^*)^3(\phi)\right),
\end{equation}
$$
\vdots 
$$
Finally, since $\left(F^*\right)^n=0$ for some $n$, then we have $\mathfrak{g}_{(n)}=0$.
\end{proof}

Now, we show when two different linear mappings give to isomorphic Lie algebras. 

\begin{theorem}
Let $V,W$ be vector spaces over $\mathbb R$, and $F\in End (V)$, $G\in End (W)$, $v\in V$, $w\in W$ such that \begin{equation}
\exists{\lambda\in\mathbb R}\quad Fv=\lambda v,\quad \exists{\gamma\in\mathbb R}\quad Gw=\gamma w.
\end{equation}
Then Lie algebras $(V^*,[\cdot ,\cdot]_{F,v})$, $(W^*,[\cdot , \cdot]_{G,w})$ are isomorphic if and only if there exists a linear bijection $\Phi:V \longrightarrow W$ such that 
\begin{equation}
\iota_{(\Phi v)}\wedge (\Phi\circ F)^*=\iota_w\wedge (G\circ\Phi )^*.
\end{equation}
 \end{theorem}
 
  \begin{proof}
  Let $\psi,\phi\in W^*$.
  If we suppose that there exists a linear bijection $\Phi:V\longrightarrow W$ such that $\iota_{(\Phi v)}\wedge (\Phi\circ F)^*=\iota_w\wedge (G\circ\Phi )^* $ we obtain
\begin{equation}
[\Phi^*\psi,\Phi^*\phi]_{F,v}=(\iota_v\wedge F^*)(\Phi^*\phi,\Phi^*\psi) =(\iota_{(\Phi v)}\wedge (\Phi\circ F)^*)(\phi,\psi)
\end{equation}
$$=(\iota_w\wedge (G\circ\Phi )^*)(\phi,\psi)=\Phi^*([\psi,\phi]_{G,w}).$$
The above gives that $\Phi^*$ is the isomorphism of Lie algebras $(V^*,[\cdot , \cdot]_{F,v})$ and $(W^*,[\cdot ,\cdot]_{G,w})$.

  On the other hand if we suppose that $(V^*,[\cdot ,\cdot]_{F,v})\cong(W^*,[\cdot , \cdot]_{G,w})$,  which means  there is a linear bijection $\Psi:W^*\rightarrow V^*$ such that 
\begin{equation}
\label{4} [\Psi\psi,\Psi\phi]_{F,v}=\Psi([\psi,\phi]_{G,w}).
\end{equation}
We denote $\Phi:=\Psi^*$. The relation (\ref{4}) implies that 
\begin{equation}
(\iota_v\wedge F^*)(\Phi^*\phi,\Phi^*\psi)=[\Phi^*\psi,\Phi^*\phi]_{F,v}=\Phi^*([\psi,\phi]_{G,w})
\end{equation}
$$=\Phi^*((\iota_w\wedge G^*)(\phi,\psi))=
  (\iota_w\wedge (G\circ \Phi)^*)(\phi,\psi),$$
for all  $\psi,\phi\in W^*$.  
 Hence the theorem is proved.
\end{proof}
Note that the assumptions of the theorem are satisfied if $F=\Phi^{-1}\circ G \circ\Phi$ and $\Phi v=w$.
Furthermore, if the mappings $F$ and $G$ have common eigenvectors, then it is $Fv=\lambda_1 v$, $Gv=\lambda_2 v$, and they act the same way on the other vectors $F\big|_{V\setminus<v>}=G\big|_{V\setminus<v>}$, the Lie algebras obtained from $F$ and $G$ are also isomorphic.
This indicates that the construction does not depend on the eigenvalue of the operator.
So we will often put it zero in the future.

We now answer the question when the linear combination of Lie brackets of form  (\ref{b1}) gives a Lie bracket.
\begin{theorem}
Let $V$ be a vector space over $\mathbb R$. If $F,G\in End (V)$, $v,w\in V $ are such that:
\begin{enumerate}
\item $v$ is an eigenvector of the map $F$,
\item $w$ is an eigenvector of the map $G$,
\item the following condition is true
\begin{equation}
\label{condition}
\iota_w\wedge [F,G]^*\wedge \iota_v
+\iota_w \wedge \iota_{Gv} \wedge F^*+ \iota_v \wedge \iota_{Fw} \wedge G^*=0.
\end{equation}
\end{enumerate}
Then 
$(V^*,[\cdot,\cdot]^\epsilon_{F,v,G,w})$, where
 \begin{equation}
  [\psi,\phi]^{\epsilon}_{F,v,G,w}= [\psi,\phi]_{F,v}+\epsilon [\psi,\phi]_{G,w}
\end{equation}
is a Lie algebra for every $\epsilon\in \mathbb R$.
\end{theorem}

\begin{proof}
 Using Jacobi identity we calculate 
\begin{equation}
 \dfrac{1}{\epsilon}\left([[\psi,\phi]^{\epsilon}_{F,v,G,w},\zeta  ]^{\epsilon}_{F,v,G,w}+
[[\zeta,\psi]^{\epsilon}_{F,v,G,w},\phi  ]^{\epsilon}_{F,v,G,w}+
[[\phi,\zeta]^{\epsilon}_{F,v,G,w},\psi  ]^{\epsilon}_{F,v,G,w}\right)
\end{equation}
$$
= \zeta(w)\left(\phi(v)\left([F,G]\right)^*(\psi)-\psi(v)\left([F,G]\right)^*(\phi)\right)
$$ $$
+\phi(w)\left(\psi(v)\left([F,G]\right)^*(\zeta)-\zeta(v)\left([F,G]\right)^*(\psi)\right)
$$ $$
+\psi(w)\left(\zeta(v)\left([F,G]\right)^*(\phi)-\phi(v)\left([F,G]\right)^*(\zeta)\right)
$$ $$
+\left(\psi(w)\phi(G(v)) - \phi(w)\psi(G(v))\right) F^*(\zeta )
+\left(\zeta(w)\psi(G(v)) - \psi(w)\zeta(G(v))\right) F^*(\phi )
$$ $$
+\left(\phi(w)\zeta(G(v)) - \zeta(w)\phi(G(v))\right) F^*(\psi )
+\left(\psi(v)\phi(F(w)) - \phi(v)\psi(F(w))\right) G^*(\zeta )
$$ $$
+\left(\zeta(v)\psi(F(w)) - \psi(v)\zeta(F(w))\right) G^*(\phi )
+\left(\phi(v)\zeta(F(w)) - \zeta(v)\phi(F(w))\right) G^*(\psi )=0
$$
for $\psi,\phi,\zeta \in V^*$.
It can be rewritten as
\begin{equation}
\bigg( \iota_w\wedge [F,G]^*\wedge \iota_v
+\iota_w \wedge \iota_{Gv} \wedge F^*+ \iota_v \wedge \iota_{Fw} \wedge G^* \bigg)(\psi, \phi, \zeta)=0.
\end{equation}
This finishes the proof.
\end{proof}

Finding the general solution of the condition (\ref{condition}) seems to be difficult, but some classes of mappings satisfying it can be easily specified. 
\begin{corollary}\label{0}
If $V$ is a vector space over $\mathbb R$ and if $F,G\in End (V)$, $ v,w\in V $ such that:
\begin{align}
& [F,G]=0, \\
& \exists {\lambda\in\mathbb R}\quad Fv=\lambda v,\\
& \exists {\gamma\in\mathbb R}\quad Gw=\gamma w,\\
&  Fw=0,\\
& Gv=0,
\end{align}
then $(V^*,[.,.]^{\epsilon}_{F,v,G,w})$  is a Lie algebra for every $\epsilon$.
\end{corollary}
 
Under these assumptions, we can tell when a sum of Lie brackets gives a nilpotent Lie algebra. 
 
\begin{theorem}
Let $V$ be a vector space over $\mathbb R$. If $F,G\in End   (V)$ are nilpotent and all the assumptions of   previous Corollary \ref{0} are fulfilled, then $(V^*,[\cdot ,\cdot]^{\epsilon}_{F,v,G,w})$  is a nilpotent Lie algebra.
\end{theorem}
 
\begin{proof}
  Let $\psi_1,...,\psi_{k+1}\in  V^*$ and $(F_1,v_1),...,(F_k,v_k)$  be pairs such that $F_i\in End(V)$, $v_i\in V$ and $F_iv_i=\lambda_i v_i$, where $\lambda_i\in\mathbb{R}$ for all $i=1,\dots, k$. We put $\epsilon=1$ which does not reduce generality. We assume that the pairs  $(F_i,v_i)$, $(F_j,v_j)$ for all $i,j\in \{1,,,,k\}$ satisfy  all the assumptions of   previous Corollary \ref{0}. Then  
\begin{equation}
[[[[\psi_1,\psi_2]^{1}_{F_1,v_1},\psi_3]^{1}_{F_2,v_2},...]^{1}_{F_{k-1},v_{k-1}},\psi_{k+1}]^{1}_{F_k,v_k}
\end{equation}
$$=((F_1\circ ...\circ F_k)^{*}\otimes \iota_{v_1}\otimes ...\otimes \iota_{v_k}-\iota_{v_1}\otimes(F_1\circ ... \circ F_k)^{*}\otimes \iota_{v_2}\otimes...\otimes \iota_{v_k})$$
$$( \psi_1, ...,     \psi_{k+1}).$$
In our case $(F_i,v_i)$ are either $(F,v)$ or $(G,w)$. From above and the assumption that $F,G $ are nilpotent the theorem is proved.
\end{proof}

\section{Examples}
\label{s4}

In this section we study the case of dimensions three and four.
We will now show that using the above  theorems we can easily obtain low dimensional Lie algebras. We use the classification from the paper \cite{n1}, see also \cite{SW}.

We will start with the three-dimensional Lie algebras.
\begin{example}
\label{example-1}
Let us take $V=\mathbb{R}^3$ with the standard basis $\{e_1, e_2, e_3\}$. 
We will show how to easily connect three-dimensional real Lie algebras with the corresponding linear mappings and their eigenvectors. 
We will restrict ourselves to the eigenvector $v=( 0,0,1)^{\top}$. Lie brackets will be defined in the space $V^*=\left(\mathbb{R}^3\right)^{\top}$ with the dual base $\{e^*_1, e^*_2, e^*_3\}$.
\begin{enumerate}
\item If we take 
\begin{equation}
F=\left(
\begin{matrix}
\lambda_1 & 0 & 0\\
0 & \lambda_2 & 0\\
0 & 0 & 0
\end{matrix}
\right),
\end{equation}
where $\lambda_1,\lambda_2\in\mathbb{R}$,
we obtain the Lie bracket of the form
\begin{equation}
 \label{b11} [\psi,\phi]_{F,v}= \lambda_1\left( \psi_1\phi_3-\psi_3\phi_1\right)e_1^{*}
+\lambda_2\left( \psi_2\phi_3-\psi_3\phi_2\right)e_2^{*},
\end{equation}
where $\psi=\psi_1 e_1^*+\psi_2 e_2^*+\psi_3 e_3^*$ and  $\phi=\phi_1 e_1^*+\phi_2 e_2^*+\phi_3 e_3^*$.
The commutator rules are following
\begin{equation}
[e_1^*,e_2^*]=0,\quad [e_1^*,e_3^*]=\lambda_1 e_1^*,\quad [e_2^*,e_3^*]=\lambda_2 e_2^*.
\end{equation}
\begin{enumerate}
\item For $\lambda_1=\lambda_2=1$, we recognize the Lie structure related to the Lie algebra $\mathfrak{g}_{3,3}$. The commutator rules for $\mathfrak{g}_{3,3}$ are $[e_1^*,e_3^*]=e_1^*$, $[e_2^*,e_3^*]=e_2^*$.
\item For $\lambda_1=-\lambda_2=1$, we recognize the Lie structure related to the Lie algebra $\mathfrak{g}_{3,4}$. The commutator rules for $\mathfrak{g}_{3,4}$ are $[e_1^*,e_3^*]=e_1^*$, $[e_2^*,e_3^*]=-e_2^*$.
\item For $\lambda_1=1$, $\lambda_2=a$, we recognize the Lie structure related to the Lie algebra $\mathfrak{g}_{3,5}^a$. The commutator rules for $\mathfrak{g}_{3,5}^a$ are $[e_1^*,e_3^*]=e_1^*$, $[e_2^*,e_3^*]=a e_2^*$.
\item For $\lambda_1= 1$, $\lambda_2=0$, we recognize the Lie structure related to the Lie algebra $\mathfrak{g}_{2,1}\oplus \langle e_2^* \rangle$. 
\end{enumerate}
\item If we take 
\begin{equation}
F=\left(
\begin{matrix}
\lambda_1 & 0 & 0\\
1 & \lambda_1 & 0\\
0 & 0 & 0
\end{matrix}
\right),
\end{equation}
where $\lambda_1\in\mathbb{R}$,
we obtain the Lie bracket of the form
\begin{align}
 \label{b12} [\psi,\phi]_{F,v} & = \left(\lambda_1\left( \psi_1\phi_3-\psi_3\phi_1\right)+\psi_2\phi_3-\psi_3\phi_2\right)e_1^{*}\\
&+\lambda_1\left( \psi_2\phi_3-\psi_3\phi_2\right)e_2^{*}.\nonumber
\end{align}
The commutator rules are following
\begin{equation}
[e_1^*,e_2^*]=0,\quad [e_1^*,e_3^*]=\lambda_1 e_1^*,\quad [e_2^*,e_3^*]=e_1^*+\lambda_1 e_2^*.
\end{equation}
\begin{enumerate}
\item For $\lambda_1=0$, we recognize the Lie structure related to the Lie algebra $\mathfrak{g}_{3,1}$. The commutator rule for $\mathfrak{g}_{3,1}$ is  $[e_2^*,e_3^*]=e_1^*$.
\item For $\lambda_1=1$, we recognize the Lie structure related to the Lie algebra $\mathfrak{g}_{3,2}$. The commutator rules for $\mathfrak{g}_{3,2}$ are $[e_1^*,e_3^*]=e_1^*$, $[e_2^*,e_3^*]=e_1^*+e_2^*$.
\end{enumerate}
\item If we take 
\begin{equation}
F=\left(
\begin{matrix}
a & -1 & 0\\
1 & a & 0\\
0 & 0 & 0
\end{matrix}
\right),
\end{equation}
where $a\in\mathbb{R}$,
we obtain the Lie bracket of the form
\begin{align}
 \label{b13} [\psi,\phi]_{F,v} & = \left( a\left( \psi_1\phi_3-\psi_3\phi_1\right)+\psi_2\phi_3-\psi_3\phi_2\right)e_1^{*}\\
&+\left( -\psi_1\phi_3+\psi_3\phi_1 + a\left( \psi_2\phi_3-\psi_3\phi_2\right)\right)e_2^{*}.\nonumber
\end{align}
The commutator rules are following
\begin{equation}
[e_1^*,e_2^*]=0,\quad [e_1^*,e_3^*]=a e_1^*- e_2^*,\quad [e_2^*,e_3^*]=e_1^*+a e_2^*.
\end{equation}
\begin{enumerate}
\item For $a\neq 0$, we recognize the Lie structure related to the Lie algebra $\mathfrak{g}_{3,7}^a$. The commutator rules for $\mathfrak{g}_{3,7}^a$ are $[e_1^*,e_3^*]=a e_1^*- e_2^*$, $[e_2^*,e_3^*]=e_1^*+a e_2^*$.
\item For $a=0$, we recognize the Lie structure related to the Lie algebra $\mathfrak{g}_{3,6}$. The commutator rule for $\mathfrak{g}_{3,6}$ are $[e_1^*,e_3^*]=-e_2^*$, $[e_2^*,e_3^*]=e_1^*$.
\end{enumerate}
\item If we take 
\begin{equation}
F=\left(
\begin{matrix}
0         & -1        & 0   \\
1         & 0         & 0    \\
0         & 0         & 0      
\end{matrix}
\right),
\end{equation}
and as the second linear mapping and its eigenvector we take
\begin{equation}
G=\left(
\begin{matrix}
0         & 0         & 1     \\
0         & 0         & 0     \\
0         & 0         & 0         
\end{matrix}
\right),\quad 
w=\left(\begin{matrix}
0 \\ 1\\  0
\end{matrix}\right),
\end{equation}
then we obtain the Lie bracket of the form
\begin{align}
 \label{b14} [\psi,\phi]_{F,v, G,w} &= [\psi,\phi]_{F,v}+[\psi,\phi]_{ G,w}=\left( \psi_2\phi_3-\psi_3\phi_2\right)e_1^{*}\\
& -\left( \psi_1\phi_3-\psi_3\phi_1\right)e_2^{*}+\left( \psi_1\phi_2-\psi_2\phi_1\right)e_3^{*}.\nonumber
\end{align}
The nonzero commutator rules are following
\begin{equation}
[e_1^*,e_2^*]= e_3^*,\quad[e_1^*,e_3^*]= -e_2^*,\quad [e_2^*,e_3^*]=e_1^*.
\end{equation}
Above we recognize the Lie structure related to the Lie algebra $\mathfrak{g}_{3,9} = \mathfrak{so}(3)$. 
\item If we take 
\begin{equation}
F=\left(
\begin{matrix}
0         & -2        & 0   \\
0         & 0         & 0    \\
0         & 0         & 0      
\end{matrix}
\right),
\end{equation}
and as the second linear mapping and its eigenvector we take
\begin{equation}
G=\left(
\begin{matrix}
1         & 0         & 0     \\
0         & 0         & 0     \\
0         & 0         & -1         
\end{matrix}
\right),\quad 
w=\left(\begin{matrix}
0 \\ 1\\  0
\end{matrix}\right),
\end{equation}
then we obtain the Lie bracket of the form
\begin{align}
 \label{b15} [\psi,\phi]_{F,v, G,w} &= [\psi,\phi]_{F,v}+[\psi,\phi]_{ G,w}=\left( \psi_1\phi_2-\psi_2\phi_1\right)e_1^{*}\\
& -2\left( \psi_1\phi_3-\psi_3\phi_1\right)e_2^{*}+\left( \psi_2\phi_3-\psi_3\phi_2\right)e_3^{*}.\nonumber
\end{align}
The nonzero commutator rules are following
\begin{equation}
[e_1^*,e_2^*]= e_1^*,\quad[e_1^*,e_3^*]= -2e_2^*,\quad [e_2^*,e_3^*]=e_3^*.
\end{equation}
Above we recognize the Lie structure related to the Lie algebra $\mathfrak{g}_{3,8} = \mathfrak{sl}(2,\mathbb{R})$. 
\end{enumerate}
\end{example}

We have obtained all three dimensional Lie algebras, see Table \ref{table 1}.
We got seven of them using one linear mapping and its eigenvector. 
To describe the other two algebras we needed two linear mappings and their eigenvectors.

In this case, after identifying $V\cong V^*\cong \mathbb{R}^3$,  the bracket (\ref{b1}) can be written as
\begin{equation}
\label{b1-1} [\psi,\phi]_{F,v}=F^{\top}\left( v\times \left(\psi\times \phi\right)\right),
\end{equation}
where we used  the vector triple product expansion
\begin{equation}
v\times \left( w \times u\right)=\left<v | u\right> w - \left< v | w\right> u,
\end{equation}
$\left<\cdot | \cdot \right>$ is the scalar product.
Additionally, it can be rewritten as
\begin{equation}
\label{b1-2} [\psi,\phi]_{F,v}( \cdot )=\left<F(\cdot)\times v |\psi\times \phi\right>,
\end{equation}
where we used the identity 
\begin{equation}
\left< v\times w | u\times t\right>=\left< v|u \right> \left< w|t \right> - \left< v|t \right> \left< w|u \right>.
\end{equation}

Recall that we have a one-to-one correspondence between the Lie algebra structure and the linear Poisson structure.
Thus on $\mathfrak{g}^*\cong \mathfrak{g}=\mathbb{R}^3$ we have the canonical Lie--Poisson structures given by the formula (\ref{bp}).
In our situation, using (\ref{b1-2}), the Poisson bracket (\ref{bp}) is written in the form 
\begin{equation}
\{f,g\}_{F,v}({\bf x})=\left<F({\bf x})\times v |\nabla f({\bf x})\times \nabla g({\bf x})\right>, \quad f,g\in C^{\infty}\left(\mathbb{R}^3\right).
\end{equation}
If we introduce a function $c$ such that it satisfies the formula $\nabla c ({\bf x}) \sim F({\bf x})\times v$ (proportionality is given with precision to the function, i.e., 
\begin{equation}
\label{Casimir}
\nabla c ({\bf x}) =l({\bf x})\left(F({\bf x})\times v\right), 
\end{equation}
where $l\in C^{\infty}\left(\mathbb{R}^3\right)$), then we recognize in this the structure of the Nambu bracket
\begin{equation}
\{f,g,c\}({\bf x})=\dfrac{\partial(f,g,c)}{\partial (x_1,x_2,x_3)}.
\end{equation}
If we treat a function $c$ as a fixed parameter we obtain from Nambu bracket a Poisson bracket for which $c$ is a Casimir function.
We illustrate it by some examples.

\begin{example}
Let us take the Lie algebra $\mathfrak{g}_{3,3}$ and let us consider the equation (\ref{Casimir}) for Casimir function in this case
\begin{equation}
\left\{\begin{array}{l}
 \dfrac{\partial c}{\partial x_1}(x_1,x_2,x_3)=x_2 l(x_1,x_2,x_3)\\
 \dfrac{\partial c}{\partial x_2}(x_1,x_2,x_3)=-x_1 l(x_1,x_2,x_3)\\
 \dfrac{\partial c}{\partial x_3}(x_1,x_2,x_3)=0
\end{array}.\right.
\end{equation}
Determining from the first equation the function $l$ and substituting into the second equation we obtain the linear first order partial differential equation
\begin{equation}
-x_1\dfrac{\partial c}{\partial x_1}(x_1,x_2,x_3)+x_2 \dfrac{\partial c}{\partial x_2}(x_1,x_2,x_3)=0.
\end{equation}
After solving, we have Casimir functions $c$ and the auxiliary function $l$
\begin{align}
& c(x_1,x_2,x_3)=\dfrac{x_2}{x_1},\\
& l(x_1,x_2,x_3)=\dfrac{1}{x_1^2},
\end{align}
see also Table \ref{table 1} (whether the integral factor appears significantly has to do with the  compatibility of the Lie brackets, see \cite{A-W, Pa}).
\end{example}

\begin{table}
\begingroup\makeatletter\def\f@size{6}\check@mathfonts
\begin{tabular}{|c|c|c|c|c|c|}
\hline   $F$                    &      $v$       & $\textrm{l({\bf x})}$               & $\textrm{Casimir}$   
& $\textrm{Name}$   \\
$G$&$w$&  &&\\
\hline $F=\left( \begin{matrix} 0 & 0 & 0\\ 1 & 0 & 0\\ 0 & 0 & 0 \end{matrix} \right)$ & 
$v=\left(\begin{matrix} 0 \\ 0 \\ 1\end{matrix}\right)$ & $ \dfrac{1}{x_1}$     &   $x_1$    & $\mathfrak{g}_{3,1}$\\
\hline $F=\left( \begin{matrix} 1 & 0 & 0\\ 1 & 1 & 0\\ 0 & 0 & 0 \end{matrix} \right)$ & 
$v=\left(\begin{matrix} 0 \\ 0 \\ 1\end{matrix}\right)$ &  $ \dfrac{1}{x_1}e^{-\frac{x_2}{x_1}}$     &   $x_1e^{-\frac{x_2}{x_1}}$     & $\mathfrak{g}_{3,2}$\\
\hline $F=\left( \begin{matrix} 1 & 0 & 0\\ 0 & 1 & 0\\ 0 & 0 & 0 \end{matrix} \right)$ & 
$v=\left(\begin{matrix} 0 \\ 0 \\ 1\end{matrix}\right)$ &  $-\dfrac{1}{x_1^2}$      & $ \dfrac{x_2}{x_1}$        & $\mathfrak{g}_{3,3}$\\
\hline $F=\left( \begin{matrix} 1 & 0 & 0\\ 0 & -1 & 0\\ 0 & 0 & 0 \end{matrix} \right)$ & 
$v=\left(\begin{matrix} 0 \\ 0 \\ 1\end{matrix}\right)$ &  $-1$     &  $x_1x_2$      & $\mathfrak{g}_{3,4}$\\
\hline $F=\left( \begin{matrix} 1 & 0 & 0\\ 0 & a & 0\\ 0 & 0 & 0 \end{matrix} \right)$ & 
$v=\left(\begin{matrix} 0 \\ 0 \\ 1\end{matrix}\right)$ &  $-\dfrac{1}{x_1^{a+1}}$      &  $\dfrac{x_1}{x_1^a}$       & $\mathfrak{g}^a_{3,5}$\\ 
\hline $F=\left( \begin{matrix} 0 & -1 & 0\\ 1 & 0 & 0\\ 0 & 0 & 0 \end{matrix} \right)$ & 
$v=\left(\begin{matrix} 0 \\ 0 \\ 1\end{matrix}\right)$ &  $2$     &  $x_1^2+x_2^2$      & $\mathfrak{g}_{3,6}$\\ 
\hline $F=\left( \begin{matrix} a & -1 & 0\\ 1 & a & 0\\ 0 & 0 & 0 \end{matrix} \right)$ & 
$v=\left(\begin{matrix} 0 \\ 0 \\ 1\end{matrix}\right)$ &  $2e^{2a \textrm{arctg}\frac{x_1}{x_2}}$     &  $(x_1^2+x_2^2)e^{2a \textrm{arctg}\frac{x_1}{x_2}}$      & $\mathfrak{g}^a_{3,7}$\\ 
\hline $F=\left( \begin{matrix} 0 & -2 & 0\\ 0 & 0 & 0\\ 0 & 0 & 0 \end{matrix} \right)$ & 
$v=\left(\begin{matrix} 0 \\ 0 \\ 1\end{matrix}\right)$ &  $1$     &  $x_1x_3+x_2^2$      & $\mathfrak{g}_{3,8}$\\ 
$G=\left( \begin{matrix} 1 & 0 & 0\\ 0 & 0 & 0\\ 0 & 0 & -1 \end{matrix} \right)$ & 
$w=\left(\begin{matrix} 0 \\ 1 \\ 0\end{matrix}\right)$&  &   & \\
\hline $F=\left( \begin{matrix} 0 & -1 & 0\\ 1 & 0 & 0\\ 0 & 0 & 0 \end{matrix} \right)$ & 
$v=\left(\begin{matrix} 0 \\ 0 \\ 1\end{matrix}\right)$ &   $2$    & $x_1^2+x_2^2+x_3^2$       & $\mathfrak{g}_{3,9}$\\ 
$G=\left( \begin{matrix} 0 & 0 & 1\\ 0 & 0 & 0\\ 0 & 0 & 0 \end{matrix} \right)$ & 
$w=\left(\begin{matrix} 0 \\ 1 \\ 0\end{matrix}\right)$&  &   & \\
\hline 
\end{tabular}
\endgroup  
\caption{\label{table 1}Linear mappings and their eigenvectors giving three dimensional Lie algebras.}
\end{table}

We now turn to the four-dimensional Lie algebras.

\begin{example}
Let us take $V=\mathbb{R}^4$ with a basis $\{e_1, e_2, e_3, e_4\}$. 
We will show how to easily connect four-dimensional real Lie algebras with the corresponding linear mappings and their eigenvectors. 
We will restrict ourselves to the eigenvector $v=( 0,0,0,1)^{\top}$. Lie brackets will be defined in the space $V^*=\left(\mathbb{R}^4\right)^{\top}$ with the dual base $\{e^*_1, e^*_2, e^*_3, e_4^*\}$.
\begin{enumerate}
\item If we take 
\begin{equation}
F=\left(
\begin{matrix}
\lambda_1 & 0         & 0         & 0\\
0         & \lambda_2 & 0         & 0\\
0         & 0         & \lambda_3 & 0 \\
0         & 0         & 0         & 0
\end{matrix}
\right),
\end{equation}
where $\lambda_1,\lambda_2,\lambda_3\in\mathbb{R}$,
we obtain the Lie bracket of the form
\begin{align}
 \label{b21} [\psi,\phi]_{F,v} &= \lambda_1\left( \psi_1\phi_4-\psi_4\phi_1\right)e_1^{*}
+\lambda_2\left( \psi_2\phi_4-\psi_4\phi_2\right)e_2^{*}\\
& +\lambda_3\left( \psi_3\phi_4-\psi_4\phi_3\right)e_3^{*},\nonumber
\end{align}
where $\psi=\psi_1 e_1^*+\psi_2 e_2^*+\psi_3 e_3^*+\psi_4 e_4^*$ and  $\phi=\phi_1 e_1^*+\phi_2 e_2^*+\phi_3 e_3^*+\phi_4 e_4^*$.
The nonzero commutator rules are following
\begin{equation}
[e_1^*,e_4^*]=\lambda_1 e_1^*,\quad [e_2^*,e_4^*]=\lambda_2 e_2^*,\quad [e_3^*,e_4^*]=\lambda_3 e_3^*.
\end{equation}
\begin{enumerate}
\item For $\lambda_1=1$, $\lambda_2=a$, $\lambda_3=b$, we recognize the Lie structure related to the Lie algebra 
$\mathfrak{g}_{4,5}^{a,b}$. The commutator rules for $\mathfrak{g}_{4,5}^{a,b}$ are $[e_1^*,e_4^*]=e_1^*$, $[e_2^*,e_4^*]=ae_2^*$, $[e_3^*,e_4^*]=b e_3^*$.
\item For $\lambda_3=0$, see Example \ref{example-1}.
\end{enumerate}
\item If we take 
\begin{equation}
F=\left(
\begin{matrix}
\lambda_1 & 0         & 0         & 0\\
1         & \lambda_1 & 0         & 0\\
0         & 1         & \lambda_1 & 0 \\
0         & 0         & 0         & 0
\end{matrix}
\right),
\end{equation}
where $\lambda_1\in\mathbb{R}$,
we obtain the Lie bracket of the form
\begin{align}
 \label{b22} [\psi,\phi]_{F,v} &= \left(\lambda_1\left( \psi_1\phi_4-\psi_4\phi_1\right)+ \psi_2\phi_4-\psi_4\phi_2\right)e_1^{*}\\
&+\left(\lambda_1\left( \psi_2\phi_4-\psi_4\phi_2\right)+ \psi_3\phi_4-\psi_4\phi_3\right)e_2^{*}
+\lambda_1\left( \psi_3\phi_4-\psi_4\phi_3\right)e_3^{*}.\nonumber
\end{align}
The nonzero commutator rules are following
\begin{equation}
[e_1^*,e_4^*]=\lambda_1 e_1^*,\quad [e_2^*,e_4^*]=e_1^*+\lambda_1 e_2^*,\quad [e_3^*,e_4^*]=e_2^*+\lambda_1 e_3^*.
\end{equation}
\begin{enumerate}
\item For $\lambda_1=0$,  we recognize the Lie structure related to the Lie algebra 
$\mathfrak{g}_{4,1}$. The commutator rules for $\mathfrak{g}_{4,1}$ are  $[e_2^*,e_4^*]=e_1^*$, $[e_3^*,e_4^*]= e_2^*$.
\item For $\lambda_1=1$,  we recognize the Lie structure related to the Lie algebra 
$\mathfrak{g}_{4,4}$. The commutator rules for $\mathfrak{g}_{4,4}$ are $[e_1^*,e_4^*]=e_1^*$, $[e_2^*,e_4^*]=e_1^*+e_2^*$, $[e_3^*,e_4^*]=e_2^*+ e_3^*$.
\end{enumerate}
\item If we take 
\begin{equation}
F=\left(
\begin{matrix}
\lambda_1 & 0         & 0         & 0\\
0         & \lambda_2 & 0         & 0\\
0         & 1         & \lambda_2 & 0 \\
0         & 0         & 0         & 0
\end{matrix}
\right),
\end{equation}
where $\lambda_1, \lambda_2\in\mathbb{R}$,
we obtain the Lie bracket of the form
\begin{align}
 \label{b23} [\psi,\phi]_{F,v} &= \lambda_1\left( \psi_1\phi_4-\psi_4\phi_1\right)e_1^{*}\\
&+\left(\lambda_2\left( \psi_2\phi_4-\psi_4\phi_2\right)+ \psi_3\phi_4-\psi_4\phi_3\right)e_2^{*}
+\lambda_2\left( \psi_3\phi_4-\psi_4\phi_3\right)e_3^{*}.\nonumber
\end{align}
The nonzero commutator rules are following
\begin{equation}
[e_1^*,e_4^*]=\lambda_1 e_1^*,\quad [e_2^*,e_4^*]=\lambda_2 e_2^*,\quad [e_3^*,e_4^*]=e_2^*+\lambda_2 e_3^*.
\end{equation}
\begin{enumerate}
\item For $\lambda_1=a$, $\lambda_2=1$,  we recognize the Lie structure related to the Lie algebra 
$\mathfrak{g}_{4,2}^a$. The commutator rules for $\mathfrak{g}_{4,2}^a$ are  $[e_1^*,e_4^*]=ae_1^*$, $[e_2^*,e_4^*]= e_2^*$, $[e_3^*,e_4^*]= e_2^*+e_3^*$.
\item For $\lambda_1=1$, $\lambda_2=0$,  we recognize the Lie structure related to the Lie algebra 
$\mathfrak{g}_{4,3}$. The commutator rules for $\mathfrak{g}_{4,3}$ are $[e_1^*,e_4^*]=e_1^*$, $[e_3^*,e_4^*]=e_2^*$.
\end{enumerate}
\item If we take 
\begin{equation}
F=\left(
\begin{matrix}
\lambda_1 & 0         & 0         & 0\\
0         & b         & -1        & 0\\
0         & 1         & b         & 0 \\
0         & 0         & 0         & 0
\end{matrix}
\right),
\end{equation}
where $\lambda_1, b\in\mathbb{R}$,
we obtain the Lie bracket of the form
\begin{align}
 \label{b24} [\psi,\phi]_{F,v} &= \lambda_1\left( \psi_1\phi_4-\psi_4\phi_1\right)e_1^{*}\\
& +\left(b\left( \psi_2\phi_4-\psi_4\phi_2\right)+ \psi_3\phi_4-\psi_4\phi_3\right)e_2^{*}\\
& +\left(-\psi_2\phi_4+\psi_4\phi_2+b\left( \psi_3\phi_4-\psi_4\phi_3\right)\right)e_3^{*}.\nonumber
\end{align}
The nonzero commutator rules are following
\begin{equation}
[e_1^*,e_4^*]=\lambda_1 e_1^*,\quad [e_2^*,e_4^*]=b e_2^*- e_3^*,\quad [e_3^*,e_4^*]=e_2^*+b e_3^*.
\end{equation}
\begin{enumerate}
\item For $\lambda_1=a$,   we recognize the Lie structure related to the Lie algebra 
$\mathfrak{g}_{4,6}^{a,b}$. The commutator rules for $\mathfrak{g}_{4,6}^{a,b}$ are  $[e_1^*,e_4^*]=ae_1^*$, $[e_2^*,e_4^*]= be_2^*-e_3^*$, $[e_3^*,e_4^*]= e_2^*+be_3^*$.
\item For $\lambda_1=0$,  see Example \ref{example-1}.
\end{enumerate}
\item If we take 
\begin{equation}
F=\left(
\begin{matrix}
0         & 0         & 0         & 0\\
0         & 0         & -1        & 0\\
0         & 1         & 0         & 0 \\
0         & 0         & 0         & 0
\end{matrix}
\right),
\end{equation}
and as the second linear mapping and its eigenvector we choose
\begin{equation}
G=\left(
\begin{matrix}
0         & 0         & 0         & 0\\
1         & 0         & 0         & 0\\
0         & 0         & 0         & 0 \\
0         & 0         & 0         & 0
\end{matrix}
\right),\quad 
w=\left(\begin{matrix}
0 \\ 0\\ 1\\ 0
\end{matrix}\right),
\end{equation}
then we obtain the Lie bracket of the form
\begin{align}
 \label{b25} [\psi,\phi]_{F,v, G,w} &= [\psi,\phi]_{F,v}+[\psi,\phi]_{ G,w}= \left( \psi_2\phi_3-\psi_3\phi_2\right)e_1^{*}\\
& +\left( \psi_3\phi_4-\psi_4\phi_3\right)e_2^{*}
-\left(\psi_2\phi_4+\psi_4\phi_2\right)e_3^{*}.\nonumber
\end{align}
The nonzero commutator rules are following
\begin{equation}
[e_2^*,e_3^*]= e_1^*,\quad [e_2^*,e_4^*]=- e_3^*,\quad [e_3^*,e_4^*]=e_2^*.
\end{equation}
Above we recognize the Lie structure related to the Lie algebra $\mathfrak{g}_{4,10}$. 
\item If we take 
\begin{equation}
F=\left(
\begin{matrix}
2         & 0         & 0         & 0\\
0         & 1         & 0         & 0\\
0         & 1         & 1         & 0 \\
0         & 0         & 0         & 0
\end{matrix}
\right),
\end{equation}
and as the second linear mapping and its eigenvector we choose
\begin{equation}
G=\left(
\begin{matrix}
0         & 0         & 0         & 0\\
1         & 0         & 0         & 0\\
0         & 0         & 0         & 0 \\
0         & 0         & 0         & 0
\end{matrix}
\right),\quad 
w=\left(\begin{matrix}
0 \\ 0\\ 1\\ 0
\end{matrix}\right),
\end{equation}
then we obtain the Lie bracket of the form
\begin{align}
 \label{b26} [\psi,\phi]_{F,v, G,w} &= [\psi,\phi]_{F,v}+[\psi,\phi]_{ G,w}\\ 
& =\left( \psi_2\phi_3-\psi_3\phi_2+2\left(\psi_1\phi_4-\psi_4\phi_1\right)\right)e_1^{*}\nonumber\\
& +\left( \psi_2\phi_4-\psi_4\phi_2+\psi_3\phi_4-\psi_4\phi_3\right)e_2^{*}
+\left(\psi_3\phi_4+\psi_4\phi_3\right)e_3^{*}.\nonumber
\end{align}
The nonzero commutator rules are following
\begin{equation}
[e_2^*,e_3^*]= e_1^*,\quad [e_1^*,e_4^*]=2 e_3^*,\quad [e_2^*,e_4^*]=e_2^*,\quad [e_3^*,e_4^*]=e_2^*+e_3^*.
\end{equation}
Above we recognize the Lie structure related to the Lie algebra $\mathfrak{g}_{4,7}$. 
\item If we take 
\begin{equation}
F=\left(
\begin{matrix}
0         & 0         & 0         & 0\\
0         & 1         & 0         & 0\\
0         & 0         & -1        & 0 \\
0         & 0         & 0         & 0
\end{matrix}
\right),
\end{equation}
and as the second linear mapping and its eigenvector we choose
\begin{equation}
G=\left(
\begin{matrix}
0         & 0         & 0         & 0\\
1         & 0         & 0         & 0\\
0        & 0         & 0         & 0 \\
0         & 0         & 0         & 0
\end{matrix}
\right),\quad 
w=\left(\begin{matrix}
0 \\ 0\\ 1\\ 0
\end{matrix}\right),
\end{equation}
then we obtain the Lie bracket of the form
\begin{align}
 \label{b27} [\psi,\phi]_{F,v, G,w} &= [\psi,\phi]_{F,v}+[\psi,\phi]_{ G,w}=\left( \psi_2\phi_3-\psi_3\phi_2\right)e_1^{*}\\
& +\left( \psi_2\phi_4-\psi_4\phi_2\right)e_2^{*}
-\left(\psi_3\phi_4+\psi_4\phi_3\right)e_3^{*}.\nonumber
\end{align}
The nonzero commutator rules are following
\begin{equation}
[e_2^*,e_3^*]= e_1^*,\quad [e_2^*,e_4^*]=e_2^*,\quad [e_3^*,e_4^*]=-e_3^*.
\end{equation}
Above we recognize the Lie structure related to the Lie algebra $\mathfrak{g}_{4,8}$. 
\item If we take 
\begin{equation}
F=\left(
\begin{matrix}
1+b         & 0         & 0         & 0\\
0         & 1         & 0         & 0\\
0         & 0         & b         & 0 \\
0         & 0         & 0         & 0
\end{matrix}
\right),
\end{equation}
and as the second linear mapping and its eigenvector we choose
\begin{equation}
G=\left(
\begin{matrix}
0         & 0         & 0         & 0\\
1         & 0         & 0         & 0\\
0        & 0         & 0         & 0 \\
0         & 0         & 0         & 0
\end{matrix}
\right),\quad 
w=\left(\begin{matrix}
0 \\ 0\\ 1\\ 0
\end{matrix}\right),
\end{equation}
then we obtain the Lie bracket of the form
\begin{align}
 \label{b28} [\psi,\phi]_{F,v, G,w} &= [\psi,\phi]_{F,v}+[\psi,\phi]_{ G,w}\\ 
& =\left( \psi_2\phi_3-\psi_3\phi_2+(1+b)\left(\psi_1\phi_4-\psi_4\phi_1\right)\right)e_1^{*}\nonumber\\
& +\left( \psi_2\phi_4-\psi_4\phi_2\right)e_2^{*}
+b\left(\psi_3\phi_4+\psi_4\phi_3\right)e_3^{*}.\nonumber
\end{align}
The nonzero commutator rules are following
\begin{equation}
[e_2^*,e_3^*]= e_1^*,\quad [e_1^*,e_4^*]=(1+b) e_1^*,\quad [e_2^*,e_4^*]=e_2^*,\quad [e_3^*,e_4^*]=be_3^*.
\end{equation}
Above we recognize the Lie structure related to the Lie algebra $\mathfrak{g}_{4,9}^b$. 
\item If we take 
\begin{equation}
F=\left(
\begin{matrix}
2a        & 0         & 0         & 0\\
0         & a         & -1         & 0\\
0         & 1         & a         & 0 \\
0         & 0         & 0         & 0
\end{matrix}
\right),
\end{equation}
and as the second linear mapping and its eigenvector we choose
\begin{equation}
G=\left(
\begin{matrix}
0         & 0         & 0         & 0\\
1         & 0         & 0         & 0\\
0         & 0         & 0         & 0 \\
0         & 0         & 0         & 0
\end{matrix}
\right),\quad 
w=\left(\begin{matrix}
0 \\ 0\\ 1\\ 0
\end{matrix}\right),
\end{equation}
then we obtain the Lie bracket of the form
\begin{align}
 \label{b29} [\psi,\phi]_{F,v, G,w} &= [\psi,\phi]_{F,v}+[\psi,\phi]_{ G,w}\\ 
& =\left( \psi_2\phi_3-\psi_3\phi_2+2a\left(\psi_1\phi_4-\psi_4\phi_1\right)\right)e_1^{*}\nonumber\\
& +\left( a\left(\psi_2\phi_4-\psi_4\phi_2\right)+\psi_3\phi_4-\psi_4\phi_3\right)e_2^{*}\nonumber\\
& +\left(-\psi_2\phi_4+\psi_4\phi_2+a\left(\psi_3\phi_4+\psi_4\phi_3\right)\right)e_3^{*}.\nonumber
\end{align}
The nonzero commutator rules are following
\begin{equation}
[e_2^*,e_3^*]= e_1^*,\quad [e_1^*,e_4^*]=2 a e_1^*,\quad [e_2^*,e_4^*]=ae_2^*-e_3^*,\quad [e_3^*,e_4^*]=e_2^*+ae_3^*.
\end{equation}
Above we recognize the Lie structure related to the Lie algebra $\mathfrak{g}_{4,11}^a$. 
\item If we take 
\begin{equation}
F=\left(
\begin{matrix}
0         & -1        & 0         & 0\\
1         & 0         & 0         & 0\\
0         & 0         & 0         & 0 \\
0         & 0         & 0         & 0
\end{matrix}
\right),
\end{equation}
and as the second linear mapping and its eigenvector we choose
\begin{equation}
G=\left(
\begin{matrix}
1         & 0         & 0         & 0\\
0         & 1         & 0         & 0\\
0         & 0         & 0         & 0 \\
0         & 0         & 0         & 0
\end{matrix}
\right),\quad 
w=\left(\begin{matrix}
0 \\ 0\\ 1\\ 0
\end{matrix}\right),
\end{equation}
then we obtain the Lie bracket of the form
\begin{align}
 \label{b210} [\psi,\phi]_{F,v, G,w} &= [\psi,\phi]_{F,v}+[\psi,\phi]_{ G,w}\\ 
& =\left( \psi_1\phi_3-\psi_3\phi_1+\psi_2\phi_4-\psi_4\phi_2\right)e_1^{*}\nonumber\\
& +\left( \psi_2\phi_3-\psi_3\phi_2+\psi_1\phi_4-\psi_4\phi_1\right)e_2^{*}.\nonumber
\end{align}
The nonzero commutator rules are following
\begin{equation}
[e_1^*,e_3^*]= e_1^*,\quad[e_2^*,e_3^*]= e_2^*,\quad [e_1^*,e_4^*]=- e_2^*,\quad [e_2^*,e_4^*]=e_1^*.
\end{equation}
Above we recognize the Lie structure related to the Lie algebra $\mathfrak{g}_{4,12}$. 
\end{enumerate}

\end{example}

In conclusion, we have obtained all four-dimensional Lie algebras, see Table \ref{table 2}.
To generate six of them, we had to use one linear mapping and its eigenvector.
The remaining six algebras required the use of two linear mappings and their eigenvectors.\\

\begin{table}
\begingroup\makeatletter\def\f@size{6}\check@mathfonts
\begin{tabular}{|c|c|c|c|c|c|c|c|}
\hline   $F$                    &      $v$       & $\textrm{Name}$ & $F$ &  $v$ & $G$      & $w$  & $\textrm{Name}$  \\
\hline $F=\left( \begin{matrix} 0 & 0 & 0 & 0\\ 1 & 0 & 0 & 0\\ 0 & 1 & 0 & 0\\ 0 & 0 & 0 & 0\end{matrix} \right)$ & 
$v=\left(\begin{matrix} 0 \\ 0 \\ 0 \\ 1\end{matrix}\right)$ &   $\mathfrak{g}_{4,1}$ 
& $F=\left( \begin{matrix} 2 & 0 & 0 & 0\\ 0 & 1 & 0 & 0\\ 0 & 1 & 1 & 0\\ 0 & 0 & 0 & 0\end{matrix} \right)$ & 
$v=\left(\begin{matrix} 0 \\ 0 \\ 0 \\ 1\end{matrix}\right)$ & 
$G=\left( \begin{matrix} 0 & 0 & 0 & 0\\ 1 & 0 & 0 & 0\\ 0 & 0 & 0 & 0 \\ 0 & 0 & 0 & 0 \end{matrix} \right)$  &  
$w=\left(\begin{matrix} 0 \\ 0 \\ 1 \\ 0\end{matrix}\right)$      & $\mathfrak{g}_{4,7}$
\\
\hline $F=\left( \begin{matrix} a & 0 & 0 & 0\\ 0 & 1 & 0 & 0\\ 0 & 1 & 1 & 0\\ 0 & 0 & 0 & 0\end{matrix} \right)$ & 
$v=\left(\begin{matrix} 0 \\ 0 \\ 0 \\ 1\end{matrix}\right)$ &       $\mathfrak{g}^a_{4,2}$
& $F=\left( \begin{matrix} 0 & 0 & 0 & 0\\ 0 & 1 & 0 & 0\\ 0 & 0 & -1 & 0\\ 0 & 0 & 0 & 0\end{matrix} \right)$ & 
$v=\left(\begin{matrix} 0 \\ 0 \\ 0 \\ 1\end{matrix}\right)$ & 
$G=\left( \begin{matrix} 0 & 0 & 0 & 0\\ 1 & 0 & 0 & 0\\ 0 & 0 & 0 & 0 \\ 0 & 0 & 0 & 0 \end{matrix} \right)$  &  
$w=\left(\begin{matrix} 0 \\ 0 \\ 1 \\ 0\end{matrix}\right)$      & $\mathfrak{g}_{4,8}$
\\
\hline $F=\left( \begin{matrix} 1 & 0 & 0 & 0\\ 0 & 0 & 0 & 0\\ 0 & 1 & 0 & 0\\ 0 & 0 & 0 & 0\end{matrix} \right)$ & 
$v=\left(\begin{matrix} 0 \\ 0 \\ 0 \\ 1\end{matrix}\right)$ &       $\mathfrak{g}_{4,3}$
& $F=\left( \begin{matrix} 1+b & 0 & 0 & 0\\ 0 & 1 & 0 & 0\\ 0 & 0 & b & 0\\ 0 & 0 & 0 & 0\end{matrix} \right)$ & 
$v=\left(\begin{matrix} 0 \\ 0 \\ 0 \\ 1\end{matrix}\right)$ & 
$G=\left( \begin{matrix} 0 & 0 & 0 & 0\\ 1 & 0 & 0 & 0\\ 0 & 0 & 0 & 0 \\ 0 & 0 & 0 & 0 \end{matrix} \right)$  &  
$w=\left(\begin{matrix} 0 \\ 0 \\ 1 \\ 0\end{matrix}\right)$      & $\mathfrak{g}^b_{4,9}$
\\
\hline $F=\left( \begin{matrix} 1 & 0 & 0 & 0\\ 1 & 1 & 0 & 0\\ 0 & 1 & 1 & 0\\ 0 & 0 & 0 & 0\end{matrix} \right)$ & 
$v=\left(\begin{matrix} 0 \\ 0 \\ 0 \\ 1\end{matrix}\right)$ &      $\mathfrak{g}_{4,4}$
& $F=\left( \begin{matrix} 0 & 0 & 0 & 0\\ 0 & 0 & -1 & 0\\ 0 & 1 & 0 & 0\\ 0 & 0 & 0 & 0\end{matrix} \right)$ & 
$v=\left(\begin{matrix} 0 \\ 0 \\ 0 \\ 1\end{matrix}\right)$ & 
$G=\left( \begin{matrix} 0 & 0 & 0 & 0\\ 1 & 0 & 0 & 0\\ 0 & 0 & 0 & 0 \\ 0 & 0 & 0 & 0 \end{matrix} \right)$  &  
$w=\left(\begin{matrix} 0 \\ 0 \\ 1 \\ 0\end{matrix}\right)$      & $\mathfrak{g}_{4,10}$
\\
\hline $F=\left( \begin{matrix} 1 & 0 & 0 & 0\\ 0 & a & 0 & 0\\ 0 & 0 & b & 0\\ 0 & 0 & 0 & 0\end{matrix} \right)$ & 
$v=\left(\begin{matrix} 0 \\ 0 \\ 0 \\ 1\end{matrix}\right)$ &        $\mathfrak{g}^{ab}_{4,5}$
& $F=\left( \begin{matrix} 2a & 0 & 0 & 0\\ 0 & a & -1 & 0\\ 0 & 1 & a & 0\\ 0 & 0 & 0 & 0\end{matrix} \right)$ & 
$v=\left(\begin{matrix} 0 \\ 0 \\ 0 \\ 1\end{matrix}\right)$ & 
$G=\left( \begin{matrix} 0 & 0 & 0 & 0\\ 1 & 0 & 0 & 0\\ 0 & 0 & 0 & 0 \\ 0 & 0 & 0 & 0 \end{matrix} \right)$  &  
$w=\left(\begin{matrix} 0 \\ 0 \\ 1 \\ 0\end{matrix}\right)$      & $\mathfrak{g}_{4,11}$
\\
\hline $F=\left( \begin{matrix} a & 0 & 0 & 0\\ 0 & b & -1 & 0\\ 0 & 1 & b & 0\\ 0 & 0 & 0 & 0\end{matrix} \right)$ & 
$v=\left(\begin{matrix} 0 \\ 0 \\ 0 \\ 1\end{matrix}\right)$ &        $\mathfrak{g}^{ab}_{4,6}$
& $F=\left( \begin{matrix} 0 & -1 & 0 & 0\\ 1 & 0 & 0 & 0\\ 0 & 0 & 0 & 0\\ 0 & 0 & 0 & 0\end{matrix} \right)$ & 
$v=\left(\begin{matrix} 0 \\ 0 \\ 0 \\ 1\end{matrix}\right)$ & 
$G=\left( \begin{matrix} 1 & 0 & 0 & 0\\ 0 & 1 & 0 & 0\\ 0 & 0 & 0 & 0 \\ 0 & 0 & 0 & 0 \end{matrix} \right)$  &  
$w=\left(\begin{matrix} 0 \\ 0 \\ 1 \\ 0\end{matrix}\right)$      & $\mathfrak{g}_{4,12}$
\\
\hline 
\end{tabular}
\endgroup  
\caption{\label{table 2}Linear mappings and their eigenvectors giving four dimensional Lie algebras.}
\end{table}

We see that in the case of building three- and four-dimensional Lie algebras (their classification), the first "layer" of linear mapping is closely related to the classification of equilibrium points for a linear homogeneous system of differential equations with constant coefficients.
We can assume that the first "layer" is related to the Jordan form of the map $F$ (more precisely with the class of such maps).

\section{General structure}
\label{s5}

In this section, we study a general situation of arbitrary finite dimension.
We answer the question how a linear mapping 
and its eigenvector or linear mappings and their eigenvectors can be assigned to a Lie algebra, 
when its structure constants are known. 
We also find the geometric formula for Casimir functions in the language of pairs $(F_i,v_i)$, where $F_i\in End(V)$ and $v_i$  is the eigenvector of the mapping $F_i$.

Let $\{e_1, \dots, e_n\}$ be a basis of $\mathfrak{g}$.
As is well known, Lie algebra structure is given by commutation relations $[e_i,e_j]=\sum_{k=1}^{n}c_{ij}^ke_k$, where $c_{ij}^k$ are structure constants.
In the first step, suppose that we have nonzero commutation relations for the basis vector  $e_n$ of the form
\begin{equation}
[e_i,e_n]=\sum_{k=1}^{n}c_{in}^k e_k, \quad i=1,\dots,n-1.
\end{equation}
We define a linear mapping $F_1$ in the following form (we identify $V\cong V^*\cong \mathbb{R}^n$)
\begin{equation}
F_1=\left(
\begin{array}{cccc|c}
c_{1n}^1         & c_{1n}^2      & \hdots        & c_{1n}^{n-1}     &  0\\
c_{2n}^1         & c_{2n}^2      & \hdots        & c_{2n}^{n-1}     & 0\\
\vdots           & \vdots        & \ddots        & \vdots           & \vdots\\
c_{n-1\;n}^1     & c_{n-1\;n}^2  & \hdots        & c_{n-1\;n}^{n-1}  & 0\\
\hline
0                & 0             & \hdots        & 0                & 0
\end{array}
\right).
\end{equation} 
We chose the eigenvector $v_1=(0,\dots, 0, 1)^{\top}$ for the mapping $F_1$ corresponds 
as the vector $e_n$. In the above matrix, the structure constants $c_{in}^n$, $i=1,\dots, n-1$ do not appear.
They will be placed in the next mappings $F_2, \dots, F_n$. To be precise, $c_{in}^n$ will appear in the mapping $F_i$.

In the second step, we have only $n-1$ basis vectors $\{e_1, \dots, e_{n-1}\}$. Again, 
we define next linear mapping  to be:
\begin{equation}
F_2=\left(
\begin{array}{cccc|c|c}
c_{1\;n-1}^1         & c_{1\;n-1}^2      & \hdots        & c_{1\;n-1}^{n-2}     & 0 & c^n_{1\; n-1}\\
c_{2\;n-1}^1         & c_{2\;n-1}^2      & \hdots        & c_{2\;n-1}^{n-2}     & 0 & c^n_{2\; n-1}\\
\vdots               & \vdots            & \ddots        & \vdots               & \vdots & \vdots\\
c_{n-2\;n-1}^1       & c_{n-2\;n-1}^2    & \hdots        & c_{n-2\;n-1}^{n-2}   & 0 & c^n_{n-2\; n-1}\\
\hline
0                    & 0                 & \hdots        & 0                    & 0 & 0\\
\hline
0                    & 0                 & \hdots        & 0                    & 0 & -c^n_{n-1\; n}
\end{array}
\right).
\end{equation} 
We chose the eigenvector $v_2=(0,\dots, 0, 1,0)^{\top}$ for the mapping $F_2$ to be the vector $e_{n-1}$ and its commutation relations are
\begin{equation}
[e_i,e_{n-1}]=\sum_{k=1}^{n}c_{in-1}^k e_k, \quad i=1,\dots,n-2.
\end{equation}

We repeat this procedure as long as it is possible until we obtain the mapping
\begin{equation}
F_{n-i+1}=\left(
\begin{array}{cccc|c|ccc}
c_{1\;i}^1         & c_{1\;i}^2      & \hdots        & c_{1\;i}^{i-1}     & 0 & c^{i+1}_{1\; i} & \hdots &   c^{n}_{1\; i}    \\
c_{2\;i}^1         & c_{2\;i}^2      & \hdots        & c_{2\;i}^{i-1}     & 0 & c^{i+1}_{2\; i}  & \hdots &   c^{n}_{2\; i} \\
\vdots               & \vdots            & \ddots        & \vdots               & \vdots & \vdots  & \ddots & \vdots\\
c_{i-1\;i}^1       & c_{i-1\;i}^2    & \hdots        & c_{i-1\;i}^{i-1}   & 0 & c^{i+1}_{i-1\; i} & \hdots &   c^{n}_{i-1\; i} \\
\hline
0                    & 0                 & \hdots        & 0                    & 0 & 0         & \hdots & 0 \\
\hline
0                    & 0                 & \hdots        & 0                    & 0 & -c^{i+1}_{i\; i+1}           & \hdots & 0 \\
\vdots               & \vdots            & \ddots        & \vdots               & \vdots & \vdots  & \ddots & \vdots\\
0                    & 0                 & \hdots        & 0                    & 0      & 0       & 0 & -c^n_{i\; n}
\end{array}
\right)
\end{equation} 
with eigenvector $v_{n-i+1} = e_i$.

Finally, after $n$ steps, we come to $e_1$ corresponding to the eigenvector $v_{n}=(1,0,\dots, 0)^{\top}$ for the mapping
\begin{equation}
F_{n}=\left(
\begin{array}{c|ccc}
 0 & 0         & \hdots & 0 \\
 \hline
 0 & -c^{2}_{1\; 2}           & \hdots & 0 \\
 \vdots & \vdots  & \ddots & \vdots\\
0      & 0       & 0 & -c^n_{1\; n}
\end{array}
\right).
\end{equation} 

This means that we can associate $n$ linear mappings $F_1,F_2,\dots,F_{n}$ and their eigenvectors $v_1,v_2,\dots, v_n$  with an $n$-dimensional Lie algebra. These mappings and their eigenvectors in turn give a Lie bracket on $\mathbb{R}^n$ 
\begin{equation}
\label{3-1}
[\psi,\phi]_{F_1, v_1,\dots, F_n,v_n}=[\psi,\phi]_{F_1,v_1}+[\psi,\phi]_{F_2,v_2}+\ldots +[\psi,\phi]_{F_{n},v_{n}}.
\end{equation}
Summarizing the above considerations, we have the following theorem.

\begin{theorem}
\label{theorem-1}
Every Lie algebra $(\mathfrak{g},[\cdot,\cdot])$ is isomorphic to the corresponding Lie algebra 
$(\mathbb{R}^n, [\cdot,\cdot]_{F_1, v_1,\dots, F_n,v_n})$.
\end{theorem}

Proof follows from writing explicitly commutation relations for $[\cdot,\cdot]_{F_1, v_1,\dots, F_n,v_n}$ and comparing them with commutation relations for $\mathfrak g$.

The isomorphism $(\mathfrak{g},[\cdot,\cdot])\cong (\mathbb{R}^n, [\cdot,\cdot]_{F_1, v_1,\dots, F_n,v_n})$ is not canonical, we can assign the linear mappings and their eigenvectors differently. 

Each bracket (\ref{b1}) in the sum (\ref{3-1}) can be rewritten in the form
\begin{equation}
\label{next}
[\psi,\phi]_{F_i, v_i}(\cdot)=\left<F_i(\cdot)\wedge v_i|\psi \otimes  \phi\right>.
\end{equation}
Using the formulas (\ref{bp}) and (\ref{next}), we can write (\ref{3-1}) in the language of Poisson brackets 
\begin{equation}
\label{f1}
\{f,g\}_{F_1, v_1,\dots, F_n,v_n}({\bf x})=\left< \sum_{i=1}^n \left(F_i({\bf x})\wedge v_i\right) |d f({\bf x})\otimes d g({\bf x})\right>, \quad f,g\in C^{\infty}\left(\mathbb{R}^n\right),
\end{equation}
where ${\bf x}=(x_1,\dots, x_n)$.
In the above formula $F_i({\bf x})\wedge v_i\in \bigwedge^2 V$ and $d f({\bf x})\wedge d g({\bf x})\in \bigwedge^2 V^{*}$.
Of course, we can explicitly assign a two-vector $\nabla  f({\bf x})\wedge \nabla  g({\bf x})\in \bigwedge^2 V$ for a two-form $d f({\bf x})\wedge d g({\bf x})$.
Since we have identified $V$ with $\mathbb{R}^n$ so we have the scalar product which is symmetric.
It extends naturally to the space of $k$-vectors $\bigwedge^k V$
\begin{equation}
\left<w|t\right>=det \left( \left<w_i|t_j\right> \right),
\end{equation}
where $i,j=1, \dots, k$ and $w=w_1\wedge \dots \wedge w_k$, $t=t_1\wedge \dots \wedge t_k$.
Finally, using the Hodge star operator $*:\bigwedge^2 V\longrightarrow \bigwedge^{n-2} V$, which has the following property
\begin{equation}
\left<w|t\right>=det \left(w\wedge *t\right), \quad w,t\in \bigwedge^2 V,
\end{equation}
formula (\ref{f1})  can be rewritten as 
\begin{equation}
\label{f2}
\{f,g\}_{F_1, v_1,\dots, F_n,v_n}({\bf x})=det \left( \nabla f({\bf x})\wedge \nabla g({\bf x})\wedge *\sum_{i=1}^n\left(F_i({\bf x})\wedge v_i\right)\right).
\end{equation}	
Summarizing the above considerations, we can present the following theorem describing Casimir functions.

\begin{theorem}
\label{theorem-2}
Casimir functions $c_i$, $i=1,\dots,k$, for the Lie algebra $(\mathbb{R}^n, [\psi,\phi]_{F_1, v_1,\dots, F_n,v_n})$ satisfy the following equation 
\begin{equation}
\label{cas-n}
\nabla c_i\wedge *\sum_{j=1}^n \left(F_j({\bf x})\wedge v_j \right)=0.
\end{equation}
\end{theorem}

\begin{proof}
The proof of the theorem follows directly from the Poisson bracket form (\ref{f2}) and the definition of the Casimir function. 
\end{proof}

If only the constant function satisfies the equation (\ref{cas-n}), then the Lie algebra has no Casimir functions.
In the case when the algebra has $n-2$ Casimir functions, they also satisfy the following equation
\begin{equation}
\nabla c_1\wedge \nabla c_2\wedge \dots \wedge \nabla c_{n-2}=l({\bf x})\left(*\sum_{i=1}^n \left(F_i({\bf x})\wedge v_i\right) \right)
\end{equation}
for some function $l\in C^{\infty}(\mathbb{R}^n)$. This function plays the role of an integrating factor.

\section*{Acknowledgments}

The second author G.J. was partially supported by National Science Centre, Poland project 2020/01/Y/ST1/00123.

\bibliographystyle{plain}

\begin{thebibliography}{aaaaaa}
\bibitem{A-L-M-Y} H. Awata, M. Li, D. Minic, T. Yoneya, {\it On the quantization of Nambu brackets}, Journal of High Energy Physics, JHEP02(2001)013, 2001.
\bibitem{Bal} Balcerzak, B., {\it Linear connections and secondary characteristic classes of Lie algebroids}, Monographs of Ludz University of Technology, Ludz University of Technology Press, Łódź 2021.
\bibitem{Bo-Fe} Boza, L.; Fedriani, E.M.; Nunez, J.; Tenorio, A.F. A Historical review of the classifications of Lie algebras. Rev. Union Mat. Argent. 2013, 54, 75–99.
\bibitem{Ch-Ho} C. Chandre, A. Horikoshi, {\it Classical Nambu brackets in higher dimensions}, arXiv:2109.13663, 2021.
\bibitem{Ci-Gr-Sch} S. Cicalo, W.A. de Graaf, C. Schneider, {\it Six-dimensional nilpotent Lie algebras}, Linear Algebra and its Applications
Volume 436, Issue 1, 2012, 163-189.
\bibitem{13} T. Courant, {\it Tangent Lie algebroids}, J. Phys. A: Math. Gen., 27, 4527-4536, 1994.
\bibitem{Cr-Mo}  F. Crespo, F.J. Molero and  S. Ferrer, {\it Poisson and integrable systems through the Nambu bracket and its Jacobi multiplier}, Journal of Geometric Mechanics, 8(2): 169-178, 2016. doi: 10.3934\/jgm.2016002
\bibitem {Do-Ja} A. Dobrogowska, G. Jakimowicz, {\it Generalization of the concept of classical r-matrix to Lie algebroids}, J. Geom. Phys. 165 (2021), 1-15.
\bibitem{AJ} A. Dobrogowska, G. Jakimowicz, {\it Tangent lifts of  bi-Hamiltonian structures},  J. Math. Phys., 58, 083505, 2017.
\bibitem {AJK1} A. Dobrogowska, G. Jakimowicz, M. Szajewska, K. Wojciechowicz, {\it Deformation of
the Poisson Structure Related to Algebroid Bracket of Differential Forms and Application
to Real Low Dimension Lie Algebras}, In: Geometry, Integrability and Quantization,
I. Mladenov, V. Pulov and A. Yoshioka (Eds), Avangard Prima, Sofia, 122 -130, 2019.
\bibitem {AJK3} Dobrogowska A., Jakimowicz G., Wojciechowicz K. (2019) {\it On some deformations of the Poisson structure associated with the algebroid bracket of differential forms}. In: Kielanowski P., Odzijewicz A., Previato E. (eds) Geometric Methods in Physics XXXVII. Trends in Mathematics. Birkhäuser.
\bibitem{A-W} Dobrogowska A., Wojciechowicz K., {\it Linear bundle of Lie algebras applied to the classification of real Lie algebras}, Symmetry 13 (2021), no. 8, 1-17.
\bibitem{Zu}  Dufour J-P.,   Zung N.T., {\it Poisson Structures and Their Normal Forms}, Birkh\"auser Verlag 2005.
\bibitem{G} de Graaf, W.A., {\it Classification of Solvable Lie Algebras}. Exp. Math. 2005, 14, 15–25.
\bibitem{Gr-2} de Graaf, W.A., {Classification of 6-dimensional nilpotent Lie algebras over fields of characteristic not 2}, Journal of Algebra, Volume 309, Issue 2,  2007,  640-653.
\bibitem{Gra-Mar} J. Grabowski and G. Marmo, {\it Remarks on Nambu-Poisson and Nambu-Jacobi brackets}, J. Phys. A: Math. Gen. 32 4239, 1999.
\bibitem{GraUrb} J. Grabowski, P. Urbanski, {\it Tangent lifts of Poisson and related structures}, J. Phys. A: Math. Gen.,  28, 6743-6777, 1995.
\bibitem{Ho-Ma}  P-M. Ho, Y. Matsuo, {\it The Nambu bracket and M-theory, Progress of Theoretical and Experimental Physics}, Volume 2016, Issue 6, June 2016, 06A104, https://doi.org/10.1093/ptep/ptw075.
\bibitem{J-R}  Jóźwikowski, M., Rotkiewicz, M., {\it Higher-order analogs of lie algebroids via vector bundle comorphisms}, 	
SIGMA 14 (2018), 135.
\bibitem{Kos} Kosmann-Schwarzbach Y., {\it Poisson Manifolds, Lie Algebroids, Modular Classes: a Survey}, Symmetry, Integrability and Geometry: Mathods and Applications, SIGMA 4,005, 2008.
\bibitem{Mac} MacCallum, M.A.H. On the classification of the real four-dimensional Lie algebras. In On Einstein’s Path; Springer: Berlin, Germany, 1999; pp. 299–317.
\bibitem{2} K. Mackenzie, {\it General theory of Lie groupoids and Lie algebroids}, London Mathematical Society Lecture Note Series, Vol 213, Cambridge University Press, Cambridge, 2005.
\bibitem{Ma-O-S} K. Mackenzie, A. Odzijewicz, A. Sliżewska, {\it Poisson geometry related to Atiyah sequences}, Symmetry Integrability Geom. Methods Appl. 14 (2018), 1-29.
\bibitem{3} F. Magri, C.  Morosi, {\it A geometrical characterization of integrable Hamiltonian systems through the theory of Poisson-Nijenhuis manifolds}, Quaderno S, Universit\'a di Milano, 19, 1984.
\bibitem{4} Ch.-M. Marle, {\it Differential calculus on a Lie algebroid and Poisson manifolds},  in The J. A. Pereira da Silva BirthdayS hrift, Textos de Mathem\'atica 32, Departamento de Mathematica, Universidade de Coimbra, Coimbra, Portugal, 83-149, 2002.
\bibitem{Mu} Mubarakzyanov, G.M. On solvable Lie algebras. Izv. Vys. Ucheb. Zaved. Mat. 1963, 1, 114–123.
\bibitem{Nambu} Nambu, Y., {\it Generalized Hamiltonian dynamics}, Physical Review. D7 (8), 2405–2412, 1973.
\bibitem{O-J-S} A. Odzijewicz, G. Jakimowicz, A. Sliżewska, {\it Banach-Lie algebroids associated to the groupoid of partially invertible elements of a $W^{*}$-algebra}, J. Geom. Phys. 95 (2015), 108-126.
\bibitem{Pa}  A. Panasyuk, {\it Compatible Lie brackets: towards a classification}, Journal of Lie Theory, {\bf 24} (2014), 561-623.
\bibitem{n1} J. Patera, R.T. Sharp, P. Winternitz, H. Zassenhaus,  {\it Invariants of real low dimension Lie algebras},   J. Math. Phys., 17, 986,  1976.
\bibitem{Po-Bo-Ne} Popovych, R.O.; Boyko, V.M.; Nesterenko, M.O.; Lutfullin, M.W. Realizations of real low-dimensional Lie algebras. J. Phys. A Math. Gen. 2003, 36, 7337. 
\bibitem{Pra} Pradines, J., {\it Th\'eorie de Lie pour les grupo\"ides diff\'erentiables. Relations entre propri\'et\'es locales et globales}, C. R. Acad. Sc. Paris, Ser. A, t. 264, 245-248, 1967.
\bibitem{Sh} A.O. Shishanin, {\it Nambu mechanics and its applications}, IOP Conf. Ser.: Mater. Sci. Eng. {\bf 468} 012029, 2018.
\bibitem{SW} L. Šnobl, P. Winternitz, {\it Classification and Identification of Lie Algebras}, CRM Monograph Series, Volume: 33, American Mathematical Society together with Centre de Recherches Mathematiques, Providence, R.I., 2014.
\bibitem{Tak} L. Takhtajan, {\it On foundation of the generalized Nambu mechanics}, Communications in Mathematical Physics, 160 (1994), 295-315. doi: 10.1007/BF02103278
\bibitem{We} A. Weinstein, {\it Symplectic  groupoids  and  Poisson  manifolds}, Bull. Amer. Math. Soc., 16, 101-103, 1987.
\bibitem{Xu} Xu, P., {\it Quantum grupoids},  Commun. Math. Phys. 216, 539-581, 2001.



\end{thebibliography}

\end{document}